\documentclass[aps, pre, twocolumn, amsmath, superscriptaddress,showkeys,showpacs]{revtex4-1}
\usepackage[utf8]{inputenc}
\usepackage[title]{appendix}
\usepackage{amsmath}
\usepackage{graphicx}
\usepackage{amssymb,color,tikz,multirow}
\usepackage{mathrsfs}
\usepackage{fontenc}
\usepackage{float}
\usepackage{amsthm}
\usepackage{hyperref}[colorlinks=true,linkcolor=blue,citecolor=blue]
\usepackage{algpseudocode}
\usepackage{algorithm}
\newtheorem{theorem}{Theorem}[section]
\newtheorem{lemma}[theorem]{Lemma}
\newtheorem{definition}{Definition}

\usepackage{amsfonts}

\begin{document}

\title{Effective Vaccination Strategies in Network-based SIR Model}

\author{Sourin Chatterjee}
\email{sc19ms170@iiserkol.ac.in}
\affiliation{Department of Mathematics and Statistics, Indian Institute of Science Education and Research, Kolkata, West Bengal 741246, India}
\author{Ahad N. Zehmakan}
\affiliation{School of Computing, The Australia National University, Canberra, Australia}

\begin{abstract}
Controlling and understanding epidemic outbreaks has recently drawn great interest in a large spectrum of research communities. Vaccination is one of the most well-established and effective strategies in order to contain an epidemic. In the present study, we investigate a network-based virus-spreading model building on the popular SIR model. Furthermore, we examine the efficacy of various vaccination strategies in preventing the spread of infectious diseases and maximizing the survival ratio. The experimented strategies exploit a wide range of approaches such as relying on network structure centrality measures, focusing on disease-spreading parameters, and a combination of both. Our proposed hybrid algorithm, which combines network centrality and illness factors, is found to perform better than previous strategies in terms of lowering the final death ratio in the community on various real-world networks and synthetic graph models. Our findings particularly emphasize the significance of taking both network structure  properties and disease characteristics into account when devising effective vaccination strategies.
\end{abstract}

\maketitle

\section{Introduction}\label{Intro}

The spread of infectious diseases has long been a major public health concern. It has caused significant deaths, comorbidity, and economic losses throughout human history, as noted in ~\cite{manriquez2021spread}. It is a very complex phenomenon, and understanding the effects needs contribution from several fields of study such as epidemiology, public health, medicine, mathematics, sociology, and computer science. Mathematical modeling, which draws on disciplines such as non-linear dynamics ~\cite{upadhyay2022combating}, graph theory ~\cite{wang2017vaccination}, and statistics ~\cite{balli2021data}, is key to understanding the negative impact of infectious diseases on the population. It can provide insights into the spread of the disease, help predict the course of epidemics, estimate the number of hospitalizations and deaths, and evaluate the effectiveness of different control measures. The emergence of global pandemics such as the COVID-19~\cite{boccaletti2020modeling}, and the H1N1 Influenza~\cite{etemad2022some} in recent years has highlighted the importance of understanding and designing measures to control the dynamics of infectious disease transmission. 

Various disease modeling approaches have been attempted with a view to understanding the dynamics of the pandemic. Compartment models are the simplest way to study epidemics, where compartmentalizing populations is done on the basis of whether the disease generates immunity or not. The SIR (Susceptible-Infected-Recovered) model, developed by Kermack and McKendrick in 1927,~\cite{kermack1927contribution} is one of the classic examples of this.  However, when the heterogeneity of the data is large and several meta-data like contact information, and age have a role to play in the disease dynamics such models would fail to approximate the general principle of the spreading process to a reasonable extent.

A network disease model is a type of mathematical model that is used to investigate the transmission of infectious diseases through networks of social contacts. In a network, individuals are represented as nodes in a graph, and the edges between them represent their interactions, such as encounters at work or in public places. The weights assigned to the edges between nodes represent the probability that an infected person will transmit the disease to their connected neighbors. The weights depend on various parameters such as the strength of the relationship (i.e., the likelihood of encounter), the speed of transmission, and the length of the infectious period. In the building of a network model, parameters relevant to disease transmission should be established. This would permit studying different intervention strategies like vaccination or quarantine, in order to curb the disease, in a much more accurate and realistic fashion.

Though testing~\cite{upadhyay2020age, ghosh2021optimal}, contact tracing~\cite{kojaku2021effectiveness} and quarantining are effective measures at the early stages of a pandemic~\cite{mandal2020model}, one of the most effective ways to prevent the spread of infectious diseases is vaccination. However, the optimal allocation of vaccines, especially with limited resources, is a complex and difficult matter. In light of COVID-19 and potential future pandemics, the availability of a vaccine prioritization method is, therefore, a critical concern, cf.~\cite{petrizzelli2022beyond}. Keeping in mind that future epidemics and other viruses may affect people differently, designing and understanding various effective vaccination strategies are very essential for the successful containment of potential pandemics in the future.  

In the case of COVID-19 vaccination, most of the governments prioritized the elderly individuals, due to the proven higher fatality rate, and the healthcare workers, due to the higher likelihood of being exposed to the virus, cf.~\cite{goldstein2021vaccinating}.
A good prioritization plan should balance the need to vaccinate vulnerable populations with the need to vaccinate those who are most likely to spread the disease. A successful vaccine allocation strategy needs to take a model of disease transmission and its epidemiological aspects, like transmission rate, mortality rate, and recovery rate into account, as well as the underlying network structure.

Most prior works have disregarded the network structure and assumed that individuals are equally likely to interact with each other, for the sake of simplicity in the analysis ~\cite{aruffo2022community}. In our approach, we take the structure of the network, capturing the connections among people, into account and study the virus spreading processes on different real-world and synthetic graph data. Furthermore, we model the spread of infectious diseases using the SIR model, taking  various factors such as transmission probability and cure rates into account. This provides a valuable test bed for evaluating and comparing different vaccination strategies.

Our main goal is to devise effective and efficient vaccination strategies in this enriched set-up. We first show that finding an optimal strategy is computationally challenging (more precisely, it is NP-hard). Thus, as prior work~\cite{khansari2016centrality}, we resort to approximation and heuristic approaches. We develop a set of vaccination strategies that consider different properties and characteristics of the network as well as the disease to efficiently allocate scarce vaccine resources. This is unlike most prior work which focuses solely on network structure\cite{petrizzelli2022beyond} or diseases characteristics~\cite{goldstein2021vaccinating}. Our strategies aim to minimize the number of deaths while maximizing the number of lives saved.

We compare our proposed strategies along with several centrality-based vaccination strategies that have received significant attention in recent years, cf.~\cite{khansari2016centrality, petrizzelli2022beyond}. Centrality-based vaccination strategies prioritize people based on their location on the network, with the goal of identifying those who have the greatest potential to infect others. Targeting these individuals can effectively reduce overall disease transmission. However, unlike our proposed hybrid approaches, these strategies rely solely on the structure of the graph and do not exploit any information regarding the spread of the virus. Through an extensive set of experiments, we evaluate the performance of various strategies for different levels of vaccination coverage. Our results highlight the strengths and weaknesses of each approach, providing valuable insights for public health officials and policymakers in their efforts to stem the spread of infectious diseases.

In short, the present work contributes to the growing literature on epidemic models and vaccination strategies by providing a robust simulation framework and introducing new effective approaches to vaccine allocation. Furthermore, our comparative analysis provides valuable insights into the effectiveness of different strategies, paving the way for more targeted and effective vaccination campaigns in the face of future epidemics and pandemics.

\textbf{Roadmap.} First, we overview some further related work in more detail in Section~\ref{Work}. Then, we provide the exact formulation of our epidemic model and the problem of maximizing the survival ratio using vaccination in Section~\ref{Model}. In section~\ref{hardness}, we prove that finding an optimal vaccination strategy is NP-hard. To facilitate the necessary grounds for introducing and experimenting with different strategies, we describe the dataset and also the parameters on which the experiments were performed in Section~\ref{Experiment}. Then in Section~\ref{Vaccination}, we develop several vaccination algorithms. Finally, the outcome of our experiments and their analysis are provided in Sections~\ref{Results} and~\ref{Discussion}, respectively.

\section{Prior Work}\label{Work}

In this section, we overview some related works on various mathematical epidemic models and different control strategies, particularly vaccination.

\subsection{Epidemic Models}
One of the oldest sets of epidemic models are the compartmental models, such as SIR~\cite{kermack1927contribution, Sarkar2020modeling}. In this approach, one assumes a homogeneous mixing of the population, which permits utilizing differential equations to describe the flow of individuals between compartments over time. Though these models are usually completely deterministic and simple to analyze, they have some fundamental shortcomings and thus various extensions of them have been introduced. By considering stochastic techniques, authors of~\cite{pajaro2022stochastic} showed that such stochastic models capture the uncertainty of the system and perform better in order to predict the spread of viral diseases. Additionally, several other techniques like agent-based modeling\cite{silva2020covid,shamil2021agent}, age-based models~\cite{upadhyay2020age, upadhyay2022combating}, network models~\cite{manriquez2021spread, petrizzelli2022beyond}, machine learning or deep learning models~\cite{wang2020prediction, zeroual2020deep} have been proven useful to model the epidemics more accurately and gain deep insights into virus spreading dynamics. 

Moreno et al.~\cite{moreno2002epidemic} proposed an epidemic model on complex networks that incorporates the SIR model. By analyzing the effect of network topology, they argued that there is a threshold of spreading parameters that determines whether an epidemic will occur or not. In~\cite{preciado2009spectral}, the authors studied a simple variant of SIR on Random Geometric Graphs and showed a result of a similar flavor. They demonstrated that if $\lambda_{\max}\le \delta/\beta$, then the virus does not spread, where $\delta$ is the recovery probability and $\beta$ is the probability of an infectious node making a susceptible neighbor infectious and $\lambda_{\max}$ is the highest eigenvalue for the adjacency matrix of the underlying network.

Researchers in the area of network science have extensively studied the ideal conditions for creating and selecting super-spreaders in various applications, such as information, virus, and fire propagation,  cf.~\cite{n2020rumor,zehmakan2019spread}. Some prior work \cite{paluch2018fast,mazza2010estimating} has attempted to address the difficulty of identifying and preventing the transmission of infections, viruses, and false information. They came to the conclusion that network density and spreading within communities are crucially associated. They observed that regardless of the size and configuration of the community, the inter-community edges play a fundamental role in the propagation of an epidemic.

In~\cite{upadhyay2022combating}, the authors used an age-structured social contact matrix which has components like - school, household, work, and others. The authors in~\cite{manriquez2021spread} have reconstructed the network from several databases and simulated the disease in the network. In~\cite{stegehuis2016epidemic}, the authors experimentally studied some dynamic processes such as bond percolation and the SIR model on different network structures. They observed that randomly shuffling the inter-community edges changes the process significantly. However, randomly distributing the edges inside each community does not change the process substantially.

\subsection{Control Strategies}

Social distancing, testing, quarantining, and vaccinations are some of the key strategies to control an epidemic. The network-based epidemic models allow the researchers to study the impact of network structure and evolution on disease spread and the effectiveness of interventions such as vaccination and social distancing~\cite{gross2006epidemic}. In~\cite{upadhyay2020age}, the authors designed age-group targeted testing strategies to identify the infected ones and quarantine them which was demonstrated to be very useful to reduce the total number of infections. Contact tracing can be highly effective in a heterogeneous network, isolating fewer nodes in total but preventing more cases as shown in~\cite{kojaku2021effectiveness}.

There have been numerous studies conducted to show the effectiveness of vaccination strategies as they generate herd immunity which essentially prevents the disease from spreading. In  ~\cite{zeng2005complexity}, the authors have shown how epidemics can be controlled and eliminated from the system by using impulsive vaccination, where the application of vaccine dose period is small compared to disease dynamics. In ~\cite{aruffo2022community}, the authors have shown that vaccine coverage of 60\% (assuming life-long immunity) along with strict measures are enough to prevent the re-emergence of the Covid-19 pandemic.  

Khansari et al. ~\cite{khansari2016centrality} compared the performance of several centrality-based and hybrid centrality-based algorithms like Betweenness-degree, Closeness-degree, and newly proposed algorithms like Katz-degree, and Radiality-degree on several random graphs models like Erd\H{o}s-R\'eyni model~\cite{erdHos1961strength}, Bar\'abasi-Albert model\cite{albert2002statistical}, and Watts-Strogatz model\cite{watts1998collective}. They evaluated the effectiveness of their method by measuring the largest connected component and how these measures reduce the largest eigenvalue of the adjacency matrix in the graph. Similarly, the authors of~\cite{petrizzelli2022beyond} showed that in different graph structures, how a vaccination strategy based on centrality measures is better than no vaccination or random vaccination.

\section{Model Description}\label{Model}

Let us first provide some basic graph definitions here.
\begin{definition}
$G=\left(V,E \right)$ is said to be an undirected graph, where  $V$ is a finite non-empty set of objects named nodes and $E$ is a collection of two-element subsets of $V$ called edges.   
\end{definition}
Let us define $n:=|V|$ and $m:=|E|$.
\begin{definition}
    For a node $v\in V$, $N\left(v\right):=\{v'\in V: \{v',v\} \in E\}$ is the \emph{neighborhood} of $v$.  Furthermore, $\hat{N}(v):=N(v)\cup \{v\}$ is the \textit{closed neighborhood} of $v$.
\end{definition}

\begin{definition}
\label{degree}
     Let $d\left(v\right):=|N\left(v\right)|$ be the \emph{degree} of $v$ in $G$. We also define $d_B(v):=|N(v)\cap B|$ for a set $B\subseteq V$.
\end{definition}

\begin{definition}
    The Adjacency matrix $A$ of an undirected graph $G$ is an $n\times n$ $0$-$1$  matrix whose columns and rows represent the nodes of the graph and edges between them are represented by the entries. A value of 1 in the matrix at position $(u, v)$ indicates that there is an edge between nodes $u$ and $v$, while a value of 0 indicates that there is no edge. The matrix $A$ is obviously symmetric.
\end{definition}

To model the temporal dynamics of the epidemic outbreak, we used a discrete-time Markovian compartmental model~\cite{kiss2015mathematics} to simulate the spread of a disease on real-world networks. Our model is a generalization of the original SIR~\cite{kermack1927contribution} which is arguably the most well-established epidemics model.
The SIR model as its name suggests covers three compartments Susceptible(S), Infectious(I), and Recovered (R). We also consider two additional compartments of Dead and Vaccinated. In our model, we represent each individual through a node, so they can be in one of the following five states:
\begin{itemize}
    \item \textit{Susceptible}: A node that is not infectious, but may become infectious once in contact with an Infectious node.
    \item \textit{Infectious}: A node that is Infectious and is capable of transmitting the disease to Susceptible nodes.
    \item \textit{Recovered}: The nodes which have been Infectious and have recovered from the disease and are no longer susceptible to re-infection.
    \item \textit{Dead}: The nodes which have been Infectious and have died from the disease.
    \item \textit{Vaccinated}: A node that is not susceptible to disease due to prior immunity against the disease by vaccination. (We assume an Infectious node cannot be vaccinated.)
\end{itemize}

\textbf{Remark.} Note that the above definitions imply that once an individual is recovered/vaccinated, they do not become infected any longer. Most of our results would hold if we relax this assumption slightly, for example by allowing a recovered/vaccinated individual to become infected with 10\% of the original probability of infection. Studying the setup where the vaccines are not highly effective, or the recovered individuals can become infectious with a large probability are out of the scope of the present study.

In our graph-based SIR model, we consider a graph $G$, where each of the $n$ nodes represents an individual, and there is an edge if the corresponding two individuals are connected. A node (i.e., individual) at any given time can be in one of the five states. Then, in each discrete time round $t$, all nodes update their state following the updating rule imposed by the virus spread dynamics. 

Let $S(t)$, $I(t)$, $R(t)$, $D(t)$, and $VC(t)$ respectively denote the set of Susceptible, Infectious, Recovered, Dead, and Vaccinated nodes in the $t$-th round of the process. Furthermore, let $N_{S(t)}(v):=N(v)\cap S(t)$, $N_{I(t)}(v):=N(v)\cap I(t)$, $N_{R(t)}(v):=N(v)\cap R(t)$, $N_{D(t)}(v):=N(v)\cap D(t)$, and $N_{VC(t)}(v):=N(v)\cap VC(t)$ for a node $v\in V$.

Starting from an initial configuration, where each node is in one of the aforementioned five states, in each discrete-time round $t\in \mathbb{N}$, all nodes simultaneously update their state in the following manner, where the \textit{infection rate} $\beta$, \textit{recovery rate} $\gamma $, and the weight functions $\omega(\cdot)$, $\omega_i(\cdot)$, $\omega_r(\cdot)$, and $\omega_d(\cdot)$ are model's parameters and are explained below:
\begin{itemize}
    \item A susceptible node $v$ becomes infectious with probability
\begin{equation}
    \label{eq-1}
        \beta \times \omega_i(v)\times \frac{\sum_{u\in N_{I(t)}(v)}\omega(\{v,u\})}{\sum_{u\in N(v)} \omega(\{v,u\})}.
\end{equation}
    \item An infectious node $v$ switches to Dead with probability $\omega_d(v)$. If this does not happen, then it switches to recovered independently with probability $\gamma \; \omega_r(v)$. Otherwise, it remains Infectious.
    \item A recovered, dead, or vaccinated node's status remains unchanged.
\end{itemize}

\begin{definition}
The weight function $\omega: B\rightarrow [a,b]$ assigns a value between $a$ and $b$ to each element in $B$. We are particularly interested in the case where $B=E$ or $B=V$.
\end{definition}

In the above description of the model we relied on the following weight functions:
\begin{itemize}
    \item Infectious function: $\omega_i:V\rightarrow [0,1]$
    \item Recovery function: $\omega_r:V\rightarrow [0,1]$
    \item Death function: $\omega_d:V\rightarrow [0,0.1]$.
    \item Edge weight function: $\omega:E\rightarrow [0,1]$
\end{itemize}

The weight $\omega_i(v)$ for a node $v$ indicates how susceptible the node is to the disease. Similarly, the weight $\omega_r(v)$ for a node $v$ indicates how well the node can recover once it is infected, and $\omega_d(v)$ for a node indicates how likely the node is going to die from the disease while being infected. So, each node $v$ has a \textit{death probability} $\omega_d(v)$. once it becomes infectious, it dies with probability $\omega_d(v)$  in each time step and survives (i.e., as usual, remains infectious and eventually recovers) with probability $1-\omega_d(v)$. These could be a function of different parameters, such as age or sex. The weight $\omega(e)$ for an edge $e=\{v,u\}$ models the probability of transmission between two nodes $v$ and $u$.

Our model is a generalization of the original SIR model~\cite{kendall1956deterministic}. If we set $G$ to be the complete graph $K_n$, let $\omega(e)=1$ for every edge $e$, and define $\omega_r(v)=1$,  $\omega_i(v)=1$, and $\omega_d(v)=0$ for  every node $v$, then the model is equivalent to the SIR model.

We observe that in our model, the process eventually reaches a configuration where there is no Infectious node, and thus no node will change its state anymore.

\section{ Problem Formulation and Inapproximability Result}
\label{hardness}
In this section, we first introduce the \textsc{Vaccination Problem}. Then, building on a reduction from the \textsc{Densest Subgraph Problem}, we prove in Theorem~\ref{hardness-thm} that our vaccination problem is NP-hard.\\

\noindent\textsc{Vaccination Problem}
\\
\textbf{Input}: A graph $G=(V_G,E_G)$, weight functions $\omega_i(v)$, $\omega_d(v)$, and $\omega_r(v)$ for every node $v\in V_G$ and $\omega(e)$ for every edge $e\in E_G$, the state of each node (e.g., Susceptible, Infectious, or Recovered), and integers $k,l'$. \\
\textbf{Output}: Is the maximum expected number of survived nodes at the end of the process equal to $l'$ if only $k$ nodes can be vaccinated? \\

\noindent\textsc{Densest Subgraph Problem}
\\
\textbf{Input}: A connected graph $H=(V_H,E_H)$ and two integers $k, l$. \\
\textbf{Output}: Is the maximum number of edges in a subgraph induced by $k$ nodes in $G$ equal to $l$. \\

\begin{theorem}[\cite{manurangsi2017almost}]
\label{dense-hardness}
The \textsc{Densest Subgraph Problem} is NP-hard.
\end{theorem}

\begin{definition}[convertor]
\label{convertor}
We convert a given graph $H=(V_H,E_H)$, with $V_H=\{v_1,\cdots,v_{n_H}\}$ and $E_H:=\{e_1,\cdots,e_{m_H}\}$, to a graph $G=(V_G,E_G)$ and the weight functions $\omega(\cdot)$, $\omega_i(\cdot)$, $\omega_d(\cdot)$, and $\omega_r(\cdot)$ in the following way. We set
\begin{itemize}
    \item $V_G:= X\cup Y\cup Z$ for $X:=\{x\}$, $Y:=\{y_1,\cdots,y_{n_H}\}$, and $Z:=\{z_1,\cdots, z_{m_H}\}$.
    \item $E_G:=\{\{x,y_i\}: 1\le i\le n_H\}\cup \{\{y_i,z_j\}: v_i\in e_j, 1\le i\le n_H, 1\le j\le m_H\}$.
    \item $\omega(e)=1$ for $e\in E_G$.
    \item $\omega_i(v)=1$ for $v\in V_G$.
    \item $\omega_r(v)=0$ for $v\in V_G$.
    \item $\omega_d(v)=0$ for $v\in X\cup Y$ and $\omega_d(v)=1$ for $v\in Z$.
\end{itemize}
We observe that $n_G:=|V_G|=n_H+m_H+1$ and $m_G:=|E_G|=n_H+2m_H$. See Figure~\ref{figure-convertor} for an example.
\end{definition}
The convertor basically adds a node $y_i$ for each node $v_i$ in $H$ and a node $z_j$ for each edge $e_j$ in $H$. Then, if $z_j$ corresponds to $e_j=\{v_i,v_{i'}\}$, it adds an edge between $y_i$ (similarly $y_{i'}$) and $z_j$. Finally, it adds a node $x$ and connects it to every node $y_i$. 

\begin{figure}[h]
  \centering
  \includegraphics[width=1\linewidth]{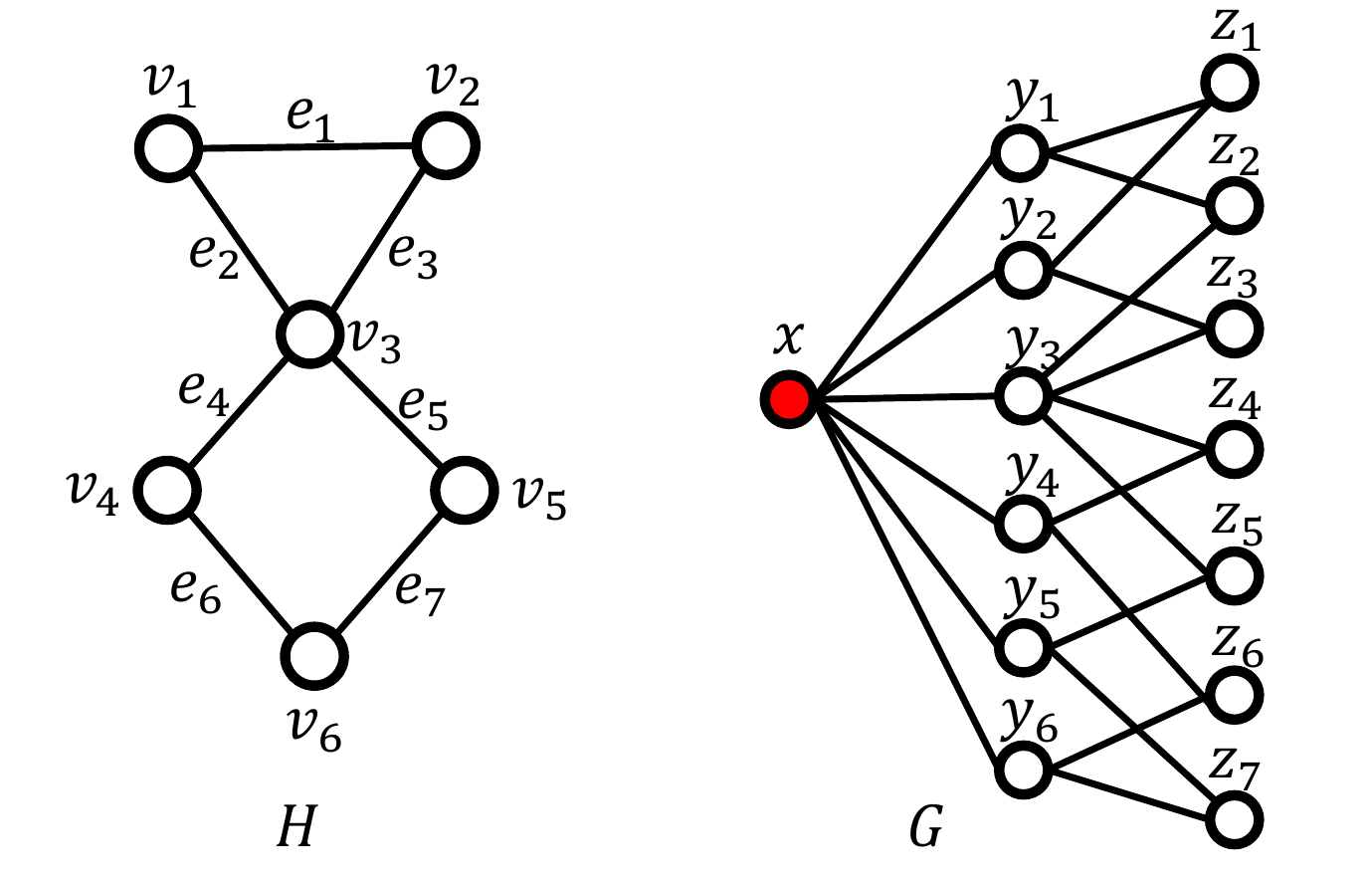}
  \caption{An example graph $H$ and the obtained graph $G$ after applying convertor in Definition~\ref{convertor}. In graph $G$, red and white nodes correspond to Infectious and Susceptible nodes, respectively.}
  \label{figure-convertor}
\end{figure}

\begin{definition}
The length of the shortest cycle (if any) in a graph $H$ is called the girth of $H$ and is denoted by $g(H)$.
\end{definition}
\begin{lemma}
\label{special-cases}
The \textsc{Densest Subgraph Problem} is polynomial time solvable in the following three cases:
\begin{itemize}
\item $k\ge n_H$.
\item $H$ is a tree (i.e., has no cycle).
\item $k< g(H)$.
\end{itemize}
\end{lemma}
\begin{proof}
If $k\ge n_H$, then all nodes can be selected, which will induce a subgraph with $m_H$ edges. Thus, the answer is YES when $l=m_H$ and NO otherwise.

If $H$ is a tree, then no subgraph induced with $k$ nodes can have more than $k-1$ edges because otherwise there exists a cycle that is in contradiction with the definition of a tree. Furthermore, any connected subgraph with $k$ nodes contains $k-1$ edges. Thus, if $l=k-1$, the answer is YES and NO otherwise.

In the third case also no induced subgraph on $k<g(h)$ nodes can have more than $k-1$ edges because otherwise, it contains a cycle of length $k$ or smaller, which is in contradiction with $k<g(H)$. Furthermore, any connected subgraph with $k$ nodes has $k-1$ edges. Thus, if $l=k-1$, the answer is YES and NO otherwise.
\end{proof}

\begin{theorem}
\label{hardness-thm}
The \textsc{Vaccination Problem} is NP-hard.
\end{theorem}
\begin{proof}
The idea is to apply a polynomial time reduction from the \textsc{Densest Subgraph Problem} to the \textsc{Vaccination Problem} and then use the hardness results from Theorem~\ref{dense-hardness}.

Let $H=(V_H,E_H)$ and integers $k,l$ be the input of the \textsc{Densest Subgraph Problem}, which do not satisfy any of the cases in Lemma~\ref{special-cases}. Define $OPT_{H,k}$ to be the maximum number of edges in a subgraph induced by $k$ nodes in $H$.  Then, we construct an instance of the \textsc{Vaccination Problem} using the converter in Definition~\ref{convertor} and let node $x$ be Infectious and all other nodes be Susceptible. (See Figure~\ref{figure-convertor}.) Define $OPT_{G,k}$ be the maximum expected number of alive nodes at the end of the process if we can vaccinate only $k$ nodes.

\textbf{Claim 1.} In the above set-up, $OPT_{G,k}=OPT_{H,k}+n_H+1$.

Let $\mathcal{A}$ be a polynomial time algorithm for the \textsc{Vaccination Problem}. Then, we claim that there is a polynomial time algorithm for the \textsc{Densest Subgraph Problem}. Note that if the input of \textsc{Densest Subgraph Problem} satisfies one of the cases in Lemma~\ref{special-cases}, then we can efficiently solve the problem. Otherwise, we construct an instance of the \textsc{Vaccination Problem} using the converter in Definition~\ref{convertor} as described above. Then, according to Claim 1, the answer to the \textsc{Densest Subgraph Problem} is YES if and only if the answer to the \textsc{Vaccination Problem} is YES for the constructed instance and $k$ and $l'=l+n_H+1$. Note that this would give a polynomial time solution to the \textsc{Densest Subgraph Problem} since the convertor clearly takes polynomial time in the input of the \textsc{Densest Subgraph Problem} and also generates an instance of polynomial time size.

It only remains to prove Claim 1. We observe that all nodes in $X\cup Y$ will never die since $\omega_d$ is equal to 0 for all these nodes. Furthermore, all these nodes, if not vaccinated, eventually become Infectious ($x$ is already Infectious from the beginning) and remain Infectious because $\omega_r(x)=0$. A node in $Z$ never becomes Infectious if it is vaccinated or both of its neighbors in $Y$ are vaccinated (recall that each node in $Z$ has exactly two neighbors in $Y$).

We first prove that $OPT_{G,k}\ge OPT_{H,k}+n_H+1$. Consider a set $S$ of $k$ nodes in $H$, which induces a subgraph with $OPT_{H,k}$ edges. Let us vaccinate the set $S':=\{y_i:v_i\in S, 1\le i\le n_H\}$. By construction, for each edge $e_j$ whose endpoints are in $S$, there is a node $z_j$ in $G$ whose both neighbors are in $S'$. Thus, all such nodes $z_j$ would never become Infectious. Furthermore, all $n_H+1$ nodes in $X\cup Y$ never die. Thus, in the end, there will be at least $OPT_{H,k}+n_H+1$ nodes alive.

Now, we prove that $OPT_{G,k}\le OPT_{H,k}+n_H+1$. Let $S$ be a node set of size $k$ in $G$ such that if $S$ is vaccinated $OPT_{G,k}$ nodes will survive. We claim that we can transform $S$ to a set $S'$ of the same size such that $S'\cap Z=S'\cap X=\emptyset$. Since the only node in $X$ (i.e., node $x$) is Infectious, it cannot be vaccinated. Define $S_Y:=S\cap Y$ and $S_Z=S\cap Z$. There must be a set $D\subset S_Y$ such that the nodes in $\{v_i:y_i\in S_Y\}$ induce a subgraph with at least $|S_Y|$ edges in $H$ because otherwise, the number of nodes that survive is at most $|X|+|Y|+|S_Y|-1+|S_Z|=1+n_H+k-1=n_H+k$. Recall that we proved $OPT_{G,k}\ge OPT_{H,k}+n_H+1$ and since $H$ has a cycle of length $k$ or smaller (note we excluded the case of $k<g(H)$ or $H$ being a tree), we have $OPT_{H,k}\ge k$ (any component including a cycle of length $g(H)\le k$ has at least $k$ edges). This implies that $OPT_{G,k}\ge k+n_H+1$, which results in a contradiction. Therefore, such a set $D$ must exist. 
Let $y_s$ be a node such that $y_s$ has a neighbor in $\{v_i:y_i\in D\}$ but is not in $D$ (such node must exist since $H$ is connected and $|D|< n_H$. The latter is true because $|D|\le |S|=k$ and we excluded the case of $k\ge n_H$). Let $z$ be a node in $S_Z$. If we vaccinate $y_s$ instead of $z$, still at least as many nodes will survive. This is because vaccinating $z$ will only save node $z$, and vaccinating $y_s$ will at least save a node $z'$ which is adjacent to $y_s$ (and maybe even more nodes). (Note that if $z'$ is vaccinated, we would have chosen $z$ to be $z'$). So far we proved there is a set of vaccinated $S'$ such that $S'\cap (X\cup Z)=\emptyset$ and will result in the survival of $OPT_{G,k}$ nodes. Since $n_H+1$ nodes in $X\cup Y$ will survive anyway. This means there are $OPT_{k,G}-n_H-1$ nodes in $Z$ whose both neighbors in $Y$ are vaccinated. By construction, the node set $\{v_i:y_i\in S', 1\le i\le n_H\}$ induces a subgraph with $OPT_{k,G}-n_H-1$ edges in $H$. This implies that $OPT_{H,k}\ge OPT_{k,G}-n_H-1$ which is equivalent to $OPT_{H,k}+n+1\ge OPT_{k,G}$. This finishes the proof.
\end{proof}

\section{Experimental Setup}\label{Experiment}

We have taken $\omega_i(v)$, $\omega_d(v)$ and $\omega_r(v)$ from a uniform random distribution, as we have not made any assumption about the recovery, infection, and death rate distributions which are a function of the studied diseases. As will be discussed, our experiments consistently support certain patterns and observations, regardless of the random choices, which makes them relevant to most setups. However, it would be interesting to study settings tailored for a particular type of disease on our model in future work. Moreover, the value of $\beta$ is set to 2, and $\gamma$ is set to 0.6 in our simulations. These factors are also dependent on the disease's type~\cite{etemad2022some, upadhyay2022combating} and different values have been utilized in various setups~\cite{silva2020covid, Sarkar2020modeling}. It is worth emphasizing that the transmission and recovery rates in our model are conceptually slightly different from previous models. More precisely, the effect of transmission and recovery rate is reduced as both parameters ($\beta$ and $\gamma$)  are multiplied by factors that lie between 0 and 1. Thus, we have chosen the values of these parameters to be aligned with the prior work, but taking the above observation into account.

For our experiments, we utilize both real-world graph data and synthetic graph models. For real-world networks, we rely on publically available data from SNAP~\cite{leskovec2012learning}. In particular, we ran our simulations on the Facebook dataset and the Twitter dataset of which some description is given below.

\begin{itemize}
    \item \textbf{Facebook} is a platform for social networking. All user-to-user connections from the Facebook network are represented in an undirected graph. An edge between two nodes $v$ and $u$ indicates that the corresponding individuals are friends. The network has $4039$ nodes and $88234$ edges. 
    \item \textbf{Twitter} is a social networking and microblogging service.  Users upload and engage with messages known as ``tweets'' on this platform. In a graph, a directed edge denotes a following relationship; for example, an edge from node $v$ to node $u$ suggests that user $v$ follows user $u$. We have converted this graph to an undirected one using the rule: an edge $e$ exists between $u$ and $v$ if there is an edge from node $v$ to node $u$ or node $u$ to node $v$. It has $81306$ nodes and $1299314$ edges. 
\end{itemize}

Most real-world networks are unweighted and one needs to introduce a meaningful procedure for weight assignment. Using the communication information of individuals on various real-world networks, the authors in~\cite{onnela2007structure} and~\cite{goyal2010learning} observed that there is a strong correlation between the number of shared friends of two individuals and their level of communication. Consequently, they proposed the usage of similarity measures, such as Jaccard-like parameters, to approximate the weights of connections between nodes. This is also aligned with the well-studied strength of weak ties hypothesis~\cite{granovetter1973strength}.
This line of research has inspired the choice of the Jaccard index in our model. Therefore, we assign the edge weights according to the Jaccard index~\cite{jaccard1901etude} in our set-up. More precisely, we set
\begin{equation}\label{eq-jaccard}
    \omega\left(\{v,u\}\right)=\frac{|\hat{N}(v)\cap \hat{N}(u)|}{N(v)\cup N(v)}.
\end{equation}
We use $|\hat{N}(v)\cap \hat{N}(u)|$ instead of $|N(v)\cap N(u)|$ in the numerator to ensure that the weight of an edge is never equal to zero.

We should emphasize that the graph data used from online social platforms do not perfectly match our use case since it is possible that two individuals are connected over an online social platform such as Facebook, but they never interact physically (and thus cannot infect each other). However, this choice can be justified by the following three reasons:
\begin{itemize}
    \item Firstly, the graph data from online social platforms are available in abundance while the graph data from physical connections between people are much more scarce. Using online social network data permits us to conduct a much more extensive and comprehensive set of experiments.
    \item It is known, cf.~\cite{costa2007characterization}, the real-world social networks regardless of their context (for example, the online social networks between people in a particular city, the physical interaction network between individuals in a certain profession, or the interest networks between fans of a particular movie genre) all share certain graph characteristics such as small diameter, scale-free degree distribution, and large clustering coefficient. Thus, while the graph data used does not match the real-world physical interactions perfectly, it still should be a very good approximation since it possess all such desired properties. This is also supported by our experiments on synthetic graph data that we explain later in this section.
    \item We assign the weights of the edges in the graph according to the Jaccard index. Consider two individuals which are connected online, but would never meet physically. In such scenarios, the two individuals perhaps are not in the small circle of friends and do not share many friends. Thus, the edge between them receives a small weight according to the Jaccard index and consequently, they are unlikely to interact in our virus-spreading process.
\end{itemize}

Another point that is worth stressing is that, of course, the Facebook and Twitter graphs used are a subgraph of the whole network. The graph data for the whole network (or even a large part of it) are not made available due to privacy reasons. Furthermore, it would not be computationally feasible to experiments on the full network even if available.

Different synthetic random graph models
have been proposed to mimic real-world social networks, cf.~\cite{costa2007characterization}.
Such models are usually tailored to possess fundamental properties consistently observed in real-world networks, such as small diameter and power-law degree distribution. We rely on the very well-established Hyperbolic Random Graph (HRG), which is a random graph model that generates complex networks with hyperbolic geometry. Nodes are embedded in hyperbolic space in the HRG model\cite{gugelmann2012random}. According to the HRG model, nodes are drawn near one another depending on their proximity to one another in geometric terms in the hyperbolic space.

We have generated HRGs such that the number of nodes and edges match with the experimented real-world networks (namely Facebook and Twitter graphs from above) using the Networkit Python Package~\cite{staudt2016networkit}.

To generate HRG, in addition to the number of nodes and edges, one needs to provide the exponent of the power-law degree distribution $b$ and the temperature $T$ as the input parameters. It is known that for HRG clustering is maximized at $T=0$, minimized at $T = \infty$, and goes through a phase transition at $T=1$, such that for $T < 1$ the graph exhibits clustering behavior whereas for $T>1$ the clustering goes to 0~\cite{krioukov2009curvature_INTERNET_TEMP}. Krioukov et al.~\cite{krioukov2009curvature_INTERNET_TEMP} demonstrated that if we embed the internet graph into hyperbolic geometry it has temperature $T=0.6$. Therefore, we also set $T=0.6$ in our setup. Moreover, we let $b = 2.5$, as it has been empirically observed that in social networks $2 \leq b \leq 3$~\cite{albert2002statistical}.

For both Facebook (and Facebook HRG) and Twitter (and Twitter HRG), we assume that initially roughly $0.5\%$ of nodes are infectious; more precisely, $20$ randomly selected nodes in the case of Facebook and $400$ randomly selected nodes in the case of Twitter. These numbers are chosen such that ensure that the virus almost surely spreads to a large part of the network; otherwise, it is not very meaningful to inject a vaccination strategy. Furthermore, for the Facebook graph (and HRG with Facebook parameters) we ran our experiments $100$ times and used the average outcome and for the Twitter graph (and HRG with Twitter parameters), $10$ repetitions were used due to its larger size. All the experiments were implemented using Python 3 and NetworkX library~\cite{hagberg2008exploring}.  

\section{Vaccination}\label{Vaccination}

The problem of finding efficient and effective vaccination strategies is very challenging. As we proved in Section~\ref{hardness}, we cannot hope to obtain a polynomial time optimal algorithm for the Vaccination Problem. Thus, we resort to approximation and heuristic approaches, as most of prior work~\cite{khansari2016centrality,erdHos1961strength}. In this section, we describe a large set of algorithms (some inspired by the classical centrality-based algorithms and some designed by us according to the virus spreading model) are presented.

We assume that we are given the budget to vaccinate up to $\alpha$ percentage of the population (or equivalently $k=\lfloor \alpha n\rfloor$ individuals) and the ultimate goal is to maximize the expected number of people alive at the end of the spread.

A natural approach is to use standard algorithms used to select the most ``influential'' nodes in a graph such as the highest degree, highest closeness, highest betweenness, cf.~\cite{petrizzelli2022beyond}. We can also consider the weighted version of these algorithms since our graph is weighted. Along with that in order to minimize the death, we have tried algorithms that consider vaccination according to higher death rates or higher death rates of neighbors. Furthermore, we propose three algorithms that rely on different formulations of the recovery, infection, and death rates. Finally, we suggest a hybrid algorithm that combines centrality measures along with disease spread parameters.

In our experiments, the $\alpha$ percentage of nodes with the highest score, according to some scoring mechanism, are vaccinated (see Algorithm~\ref{alg}). Thus, below, we simply need to define what score function is used in each strategy. For example, in Degree algorithm, the score of a node is its degree. We should emphasize that none of our algorithms assumes any knowledge of the state of the network (i.e., which nodes are Infectious/Recovered/Dead). In other words, all algorithms are ``source-agnostic''. 

\begin{algorithm}[H]
\label{alg}
\caption{Find Nodes to Vaccinate}
\label{algo:get_top_alpha_fraction}
\begin{algorithmic}[1]
\State Calculate \textsc{Score}$(v)$ for each node $v$. 
\State Sort all nodes according to \textsc{Score} in descending order.
\State Find the list $L$ of $\lfloor \alpha n\rfloor$ nodes with the highest \textsc{Score}.
\State \textbf{return} $L$.

\end{algorithmic}
\end{algorithm}

All the $16$ vaccination strategies, that have been put to test to maximize the final number of alive nodes are listed below along with their description:

\begin{enumerate}
    \item \textbf{Random}: This algorithm randomly chooses nodes for vaccination (i.e., assigns a random score to each node). 
    \item \textbf{Degree}: It measures the number of nodes to which that node is connected (as defined in Definition~\ref{degree}): \[d\left(v\right):=|N\left(v\right)|.\]

    (For example, in this algorithm, \textsc{Score}$(v)=d(v)$ and the $\lfloor \alpha n\rfloor$ nodes with the highest degree are vaccinated.)
    
    \item \textbf{Weighted Degree}: For a node $v$, we measure the sum of the weights of all its adjacent edges.  \[wd (v) := \sum_{u\in N(v)} \omega(\{v,u\}).\]
    
    \item \textbf{Eigenvector}: A node's significance in a network can be determined by looking at how it is connected to the other significant nodes in the network. In other words, the sum of the centralities of a node's neighbors defines its centrality. More precisely, given the adjacency matrix $A$ of a graph $G=(V,E)$, we define
    \[x(v) := \frac{1}{\lambda} \sum_{u\in V} A_{u,v} x(u)\]
    Then, $X=[x(v_1),x(v_2),...,x(v_n)]^T$ is the solution of the equation $AX = \lambda X$ and the $i$-th component of $X$ will give eigenvector centrality score of node $v_i$.
    
    \item \textbf{Weighted Eigenvector}: Here, the significance is weighted by edge-weight, given by the function $\omega$. To that end, we use the weighted adjacency matrix which is the same as the original one except that we use the weight $\omega(\{v,u\})$ when there is an edge between $v$ and $u$ instead of $1$.
    
    \item \textbf{Closeness}: Based on the notion that a node is important if it is close to other nodes in the network, closeness centrality is a measure of a node's relevance in a network. The inverse of the sum of the shortest distances between a node and every other node in the network is then used to establish a node's centrality. Formally, the closeness centrality of a node is given by:
\[c(v) := \frac{n-1}{\sum_{u\neq v} d(u,v)}\]
where $d(u,v)$ is the length of the shortest path between $u$ and $v$, disregarding the weights.
    
    \item \textbf{Weighted Closeness}: In the weighted closeness, the distance between two nodes is adjusted with proper weights. As higher edge weights imply two nodes being closer, we have negated the edge weights from 1 to compute the weighted closeness centrality score. So, while measuring $d(u,v)$ in place of all the edges being 1, the new weights are taken. 

    \item \textbf{Betweenness}: Betweenness centrality is a measure of a node's importance in a network based on the premise that a node is important if it lies on many shortest paths between other nodes in the network. The number of overlaps with the shortest paths between pairs of nodes is then used to determine how central a node is:
\[b(v) := \sum_{s \neq v \neq u} \frac{\sigma_{su}(v)}{\sigma_{su}}\]
where $\sigma_{su}$ the total number of shortest paths from node $s$ to node $u$ and $\sigma_{su}(v)$ is the number of those paths that pass through $v$.

    \item \textbf{Weighted Betweenness}: Similarly, in the case of weighted betweenness centrality we have considered the weighted shortest paths where the weight of an edge $\{v,u\}$ is set to $1-\omega(\{v,u\})$.
    \item \textbf{Death}: The score of a node $v$ is equal to $\omega_d(v)$. Thus, nodes with the highest death rate are vaccinated.
    \item \textbf{Neighbors' Death}: If a particular node $v$ is infected then Susceptible nodes in its surroundings (i.e, $u \in N_{S(t)}(v)$) have a chance of becoming Infectious and subsequently would die according to their death probability. Hence, we try to identify the nodes whose neighbors have a higher death rate. So, we pick nodes with the highest value of
\[
nd (v) := \sum_{u\in N(v)} \omega_d(u).
\]
    \item \textbf{Weighted Neighbors' Death}: In this algorithm, we have modified the aforementioned algorithm with the weights on the edges between two nodes as they determine the probability of disease transmission. So, we use
\[
wnd (v) := \sum_{u\in N(v)} (\omega(\{v,u\}) \cdot \omega_d(u)).
\]

\item \textbf{Expected Fatality 1}: For a node $v$ let us define the \textit{expected fatality 1} of node $v$ to be
\[
ef_1(v) := \sum_{u\in N(v)}\frac{\omega(\{v,u\})\cdot \omega_d(u)}{\sum_{w\in N(u)}\omega(\{w,u\})}+\omega_d(v)
\]
Recall from Equation~(\ref{eq-1}) that the probability of a Susceptible node $u$ becoming Infectious is proportional to 
$$\frac{\sum_{v\in N_{I(t)}(u)}\omega(\{u,v\})}{\sum_{v\in N(u)} \omega(\{u,v\})}.$$

Thus, the contribution of a node $v$ to each neighbor $u$'s infection probability is in the form $\frac{\omega(\{v,u\})}{\sum_{w\in N(u)}\omega(\{w,u\})}$. We multiply that with the death probability of $\omega_d(u)$. Furthermore, we add the death probability of node $v$ itself to the sum as well. Overall, $ef_1(v)$ is meant to account for the expected death node $v$ that could potentially cause in its closed neighborhood once it is Infectious.

\item \textbf{Expected Fatality 2}: For a node $v$ \textit{Expected Fatality 2} is defined to be:
\begin{align*}
    ef_2(v)  := \sum_{u\in N(v)}\frac{\omega(\{v,u\})\cdot \omega_d(u)}{\sum_{w\in N(u)}\omega(\{w,u\})}\\ +1-\omega_d(v)-\gamma \cdot \omega_r(v).
\end{align*}
This is the same as $ef_1(v)$, but we replace $\omega_d(v)$ by $1-\omega_d(v)-\gamma \cdot \omega_r(v)$, which is the probability that node $v$ remains Infectious. This might be relevant since the longer it remains Infectious (without becoming recovered/dead), it could potentially infect/kill more nodes. 

\item \textbf{Expected Fatality 3}: We define \textit{Expected Fatality 3} of node $v$ to be :
\begin{align*}
 ef_3(v) & := \\ &\sum_{u\in N(v)}\frac{\omega(\{v,u\})\cdot \omega_d(u)\cdot \omega_i(u) \cdot (1-\omega_d(v))}{\sum_{w\in N(u)}\omega(\{w,u\})}.
\end{align*}

This is again conceptually similar to $ef_1(v)$; however, we multiply by $\omega_i(u)\cdot (1-\omega_d(v))$. The probability $\omega_i(u)$ accounts for neighbor $u$ actually becoming Infectious. The probability $1-\omega_d(v)$ emphasizes the importance of the spreader node $v$ not dying and continuing to spread.

\item \textbf{ Hybrid Algorithm}: In this algorithm, we give importance to network structure as well as the virus spreading model parameters. We rank nodes according to their rankings in the Betweenness score, $b(v)$, and Expected Fatality 3, $ef_3(v)$, and give both scores equal weights.   
\end{enumerate}

\section{Simulation Results}\label{Results}

For each algorithm, after vaccinating the nodes, we have run the epidemic simulation and recorded the final number of deaths. We have found that some vaccination strategies like Random, Eigenvector, and Death do not perform well in any of the setups. On the other hand, Betweenness, Expected Fatality 3, and Hybrid Algorithm perform well in most of the scenarios. These results are summarized in Figures~\ref{fig:fb} and~\ref{fig:tw}.

\begin{definition}
    For a graph $G$, the survival ratio is defined as the proportion of the nodes that have not died during the epidemic out of the number of initial nodes in the graph. 
\end{definition}
For each of the studied networks, the performance of different vaccination policies is discussed below in more detail:

\begin{itemize}
    \item \textbf{Facebook}: When no vaccination is applied and we let the disease spread we found the survival ratio to be 0.817. From Figure~\ref{fig:fb}, we can see that, in $5\%-10\%$ vaccination, Closeness, Betweenness, and Weighted Betweenness performed very well. From $15\%-60\%$ vaccination, the Hybrid Algorithm is the best performer. Also after $20\%$ vaccination Expected Fatality 2 performed very well. After $50\%$ vaccination we found Degree, Betweenness, Weighted Betweenness, and Expected Fatality 3 to perform well. Throughout, the Eigenvector, Weighted Eigenvector, and Random performed poorly. The performance of the Death algorithm increased significantly towards higher vaccination percentages. 
    
    \item \textbf{Facebook HRG}: In case of no vaccination, we found the survival ratio to be 0.877. In the $5\%-10\%$ vaccination range, Betweenness and Weighted Betweenness performed very well. The hybrid algorithm performed extremely well from $15\%$ vaccination. The performance of Closeness improves significantly after $35\%$ vaccination, while Expected Fatality 2 also starts to perform better after $40\%$ vaccination. Towards $55\%-60\%$ vaccination, almost every algorithm performs very well except Random, Weighted Eigenvector, and Death.
    
 \begin{figure*}[ht]
    \centering
    \includegraphics[scale=0.65]{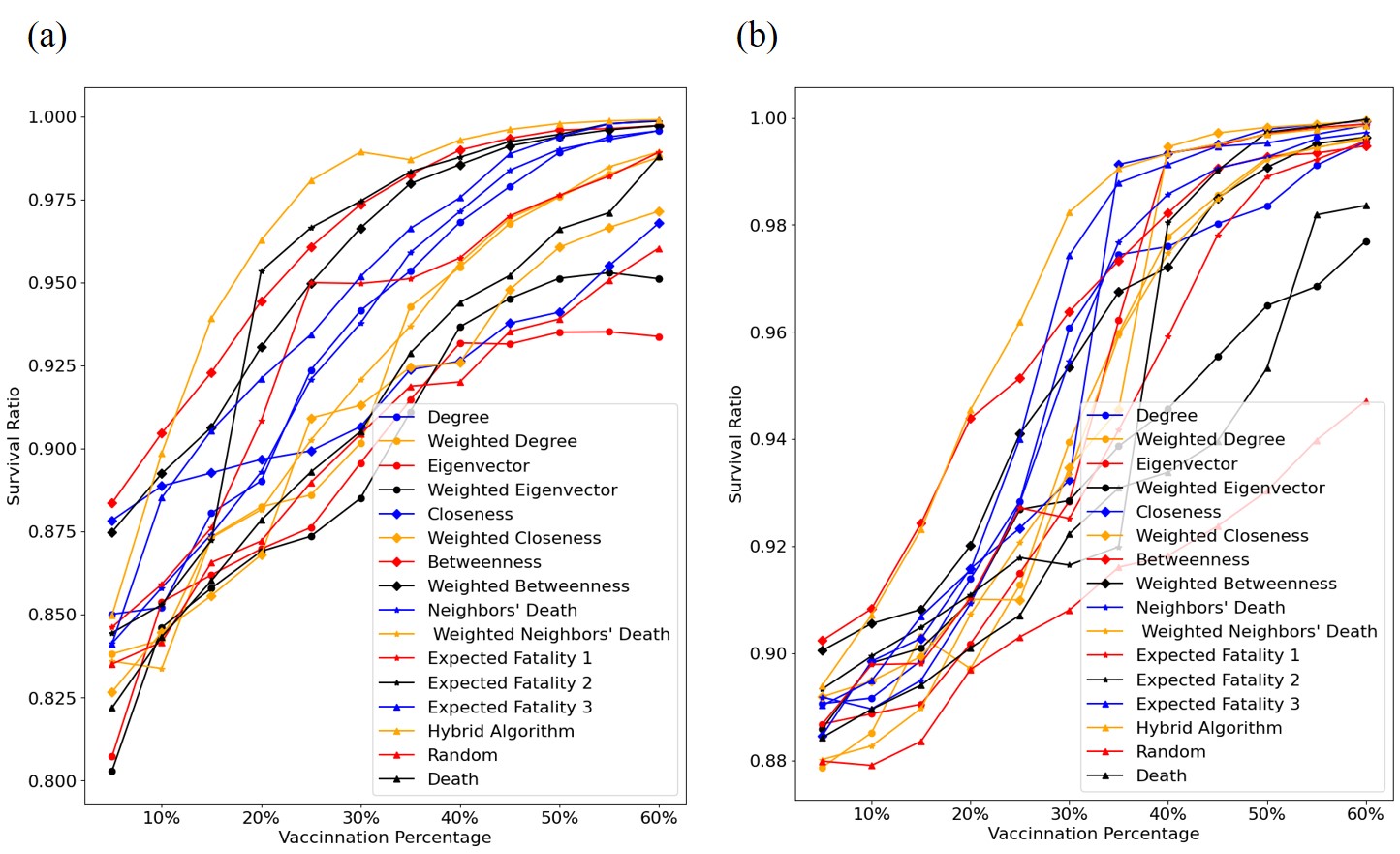}
    \caption{Survival Ratio (along the y-axis) for different percentages of vaccination (along the x-axis) according to sixteen different algorithms in (a) Facebook and (b) Facebook HRG network.}
    \label{fig:fb}
\end{figure*}   

    \item \textbf{Twitter}: Without any vaccination, 76.8\% nodes survived in the epidemic. The hybrid algorithm outperformed other algorithms in lower vaccine percentages as well as in higher vaccination percentages. From 40\% we find Degree, Weighted Degree, Betweenness, Weighted Betweenness, Neighbors' Death, Weighted Neighbors' Death, Expected Fatality 1, Expected Fatality 2, Expected Fatality 3, and Hybrid Algorithm to perform almost equal to each other with a very high survival ratio. But, the Random and Death algorithms did not perform well.  
    
    \item \textbf{Twitter HRG}: 66.4 \% of nodes survived the epidemic when there was no vaccination. Initially, in the 5\%-10\% range Hybrid algorithm, Neighbors' Death and Betweenness performed very well. Although the Hybrid Algorithm is not consistently the best performer across all vaccine percentages, it performed very well and ranked among the top 3 algorithms in terms of survival ratio. We observe that Weighted Degree performs best in the 20\%-30\% vaccination range. Onwards, 40\% vaccination range, we see Degree, Weighted Degree, Weighted Eigenvector, Betweenness, Neighbors' Death, Weighted Neighbors' Death, Expected Fatality 2, Expected Fatality 3, and Hybrid algorithm to perform almost similarly. On the other hand, the algorithms like Random, Eigenvector, and Death performed very badly. 
    
\end{itemize}

\begin{figure*}[ht]
    \centering
    \includegraphics[scale=0.65]{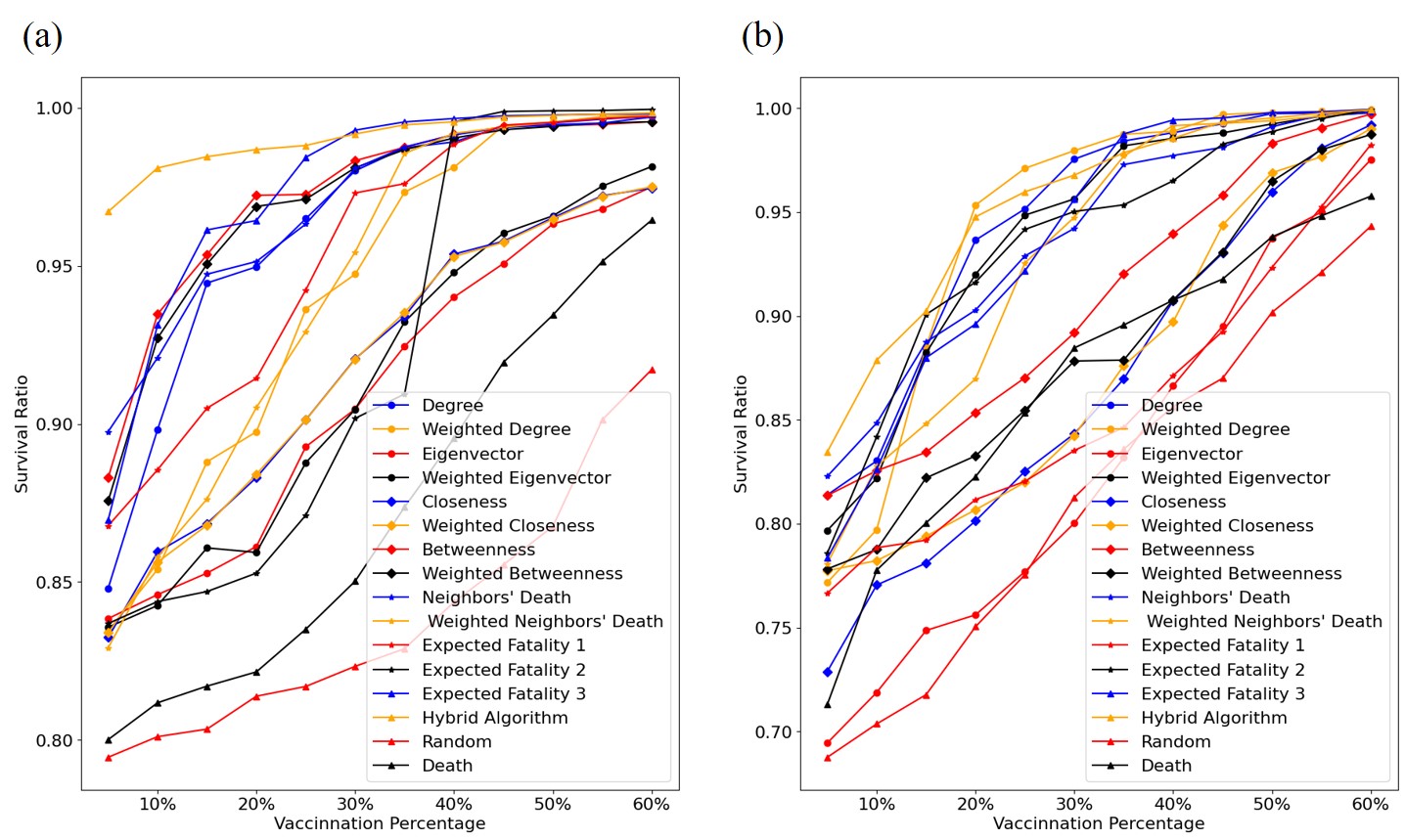}
    \caption{Survival Ratio (along the y-axis) for different percentages of vaccination (along the x-axis) according to sixteen different algorithms in (a)Twitter and (b) Twitter HRG network.}
    \label{fig:tw}
\end{figure*}

From, the results we can say fairly that algorithms almost behave similarly in real networks and their HRG counterparts. Though the performance of the Weighted Eigenvector was not very good in other networks, it performed well on Twitter HRG. We found Expected Fatality 2 to perform very well consistently on Facebook and Twitter HRG and in the other two scenarios, we found a sharp increase in performance around the 35\%-40\% vaccination range. Moreover, we found that weighted algorithms mostly behaved similarly to the non-weighted ones. It is also expected that vaccinating important nodes is far better than Random vaccination that's why it performed very poorly.  Also, vaccinating people with only the highest death rate allows the disease to spread to a larger population indicating other network and model parameters need to be considered for an effective strategy. The standard deviation of each result is mentioned in Appendix~\ref{std-appendix}. 

\section{Discussion}\label{Discussion}

We have provided a thorough review of several vaccination tactics for limiting simulated outbreaks on real-world and synthetic graph data in this work. In order to determine the most effective method for stopping the spread of diseases, our research set out to build successful strategies and compare them with classical centrality measure-based and death rate-based algorithms. Our proposed algorithms turned out to be more successful at greater vaccination rates. However, closeness and betweenness measurements of centrality generally fared pretty well. These findings show that it is essential for effectiveness that these centrality metrics be included in vaccination plans. It is very important to note that these results are robust as experiments were performed multiple times, so the effect of randomness is minimal. Moreover, our vaccination strategy does not depend upon the information about the initially infected nodes. 

One significant contribution is the hybrid algorithm's higher performance when our best model-based approach and the betweenness centrality measure were combined. In most cases, this hybrid algorithm performed better than other approaches, demonstrating the potential advantages of combining various methods into a single, coherent approach. Our findings are in line with other research that discovered centrality indicators to be useful in immunization tactics~\cite{kitsak2010identification}. The hybrid algorithm's effectiveness gives credibility to the fact that integrating several approaches can result in better epidemic containment~\cite{salathe2010dynamics, lee2012exploiting}, especially if such approaches take into account both graph structure properties and fundamental characteristics of the virus spreading process.

An interesting by-product of our experiments is the observation that based on the epidemics spreading and from the performance of several vaccination strategies, we can say that the results in the synthetic graphs are quite similar to that of the real networks which approve the idea that the HRG graph models real-world networks up to a very good extent~\cite{gugelmann2012random}. 

Although our study has shown how successful the hybrid algorithm is, there are certain drawbacks to be aware of. Firstly, due to a lack of information regarding the disease parameters, $\omega_i(v)$, $\omega_d(v)$, and $\omega_r(v)$ were taken from a uniform distribution. However, in reality, these weights depend on various parameters such as age and sex and may vary substantially across different diseases.
Another problem may be all these physiological factors are impossible to know beforehand, but that can be treated by associating them with chances of other diseases, age, and other observable parameters. Secondly, while real-world contact networks are frequently dynamic and ever-evolving, our study concentrated on static networks. Investigating how our suggested tactics perform in dynamic networks may offer insightful tips for developing more flexible and reliable immunization systems. Moreover, there might be differences between the real-world connection of people through which the epidemic spreads and social networks, cf.~\cite{zhang2014comparison}. Thirdly, we have not considered re-infection or infection even after vaccination, which may be true in real-world scenarios~\cite{sciscent2021covid}. To confirm the applicability of our findings in various setups, additional study is required.

The viability of applying these algorithms in realistic contexts must also be taken into account. The total effect of these tactics can be considerably impacted by elements including vaccine availability, logistical difficulties, and public acceptability of immunization programs. Future research should therefore focus on incorporating these factors into the formulation and assessment of vaccination regimens. 

In conclusion, our research has shown the possibility of integrating different approaches in creating more potent vaccination schemes. In particular, the Hybrid vaccination algorithm showed encouraging outcomes in containing simulated epidemics. It is essential that we use the power of computational tools and network analysis to develop creative public health protection policies as we continue to face the threat of infectious illnesses. Taken together with testing, contact tracing, and quarantining this method can be proven powerful in containing future pandemics. 

\section*{Authors' Contribution}
\textbf{Sourin Chatterjee}: Conceptualization, Methodology, Software, Validation, Formal analysis, Investigation Visualization, Writing. \textbf{Ahad N. Zehmakan}: Conceptualization, Methodology, Validation, Formal analysis, Writing. 

\section*{Declaration of Competing Interest}
The authors declare that they have no known competing financial interests or personal relationships that could have appeared to influence the work reported in this paper.

\section*{Acknowledgment}
SC thanks Rudra Mukhopadhyay for his help and valuable comments on the code. 

\begin{appendices}
\section{Standard Deviation of Result} 
\label{std-appendix}
For different percentages of vaccination, we have reported the standard deviation of the Hybrid algorithm in Table~\ref{table:sd}. These results are in terms of the number of nodes where Facebook and Facebook HRG have 4039 nodes and Twitter and Twitter HRG have 81306 nodes. As one might expect the standard deviation decreases as the vaccination percentage increases. Similar behavior was observed for the other 15 algorithms; thus, to avoid redundancy, they are not included here.

\begin{table}[!ht]
\label{table:sd}
    \centering
    \caption{Standard Deviation for Hybrid Algorithm}
    \begin{tabular}{|l|l|l|l|l|}
    \hline
        Vaccine (\%) & Facebook & Facebook HRG & Twitter & Twitter HRG \\ \hline
        5\% & 115.4 & 230.45 & 4200.17 & 10860.67 \\ \hline
        10\%  & 124.21 & 196.36 & 222.69 & 2318.4 \\ \hline
        15\%  & 102 & 220.31 & 326.81 & 5984.89 \\ \hline
        20\%  & 78.41 & 129.79 & 172.26 & 5034.48 \\ \hline
        25\%  & 36.78 & 124.49 & 67.87 & 4333.26 \\ \hline
        30\%  & 22.83 & 83.51 & 89.28 & 950.22 \\ \hline
        35\%  & 20.96 & 29.53 & 35.76 & 235.28 \\ \hline
        40\% & 12.78 & 71.72 & 20.91 & 786.84 \\ \hline
        45\% & 9.39 & 73.59 & 23.09 & 666.41 \\ \hline
        50\% & 4.63 & 53.5 & 9.24 & 548.95 \\ \hline
        55\% & 3.37 & 4.63 & 9.3 & 229.48 \\ \hline
        60\% & 2.31 & 6.96 & 38.34 & 5.32 \\ \hline
    \end{tabular}
\end{table}


\end{appendices}

\section*{Data Availability}
The data supporting this study are available on request.

\bibliographystyle{apsrev4-1}

\end{document}